\documentclass{article}
\usepackage{amsfonts,amsmath,amssymb}
\usepackage{color}
\usepackage{hyperref}
\newcommand{\Keywords}[1]{\par\noindent
{\small{\em Keywords\/}: #1}}

\def\N{\mathbb{N}}

\def\Q{\mathbb{Q}}
\def\R{\mathbb{R}}

\def\K{\mathbb{K}}
\def\F{\mathbb{F}}

\let\bfs\bfseries
\def\parag#1{\medskip\noindent{\bfs #1}\ \ }
\newenvironment{proof}{\parag{Proof}}{\null\hfill$\Box$\par\medskip}

\DeclareSymbolFont{lasy}{U}{lasy}{m}{n}
\SetSymbolFont{lasy}{bold}{U}{lasy}{b}{n}
\let\Box\undefined
\DeclareMathSymbol\Box{\mathord}{lasy}{"32}

\def\step{\mathop{\mathrm{step}}}

\newtheorem{defi}{Definition}[section]
\newtheorem{prop}{Proposition}[section]
\newtheorem{teor}{Theorem}[section]

\newtheorem{ejem}{Example}[section]
\newtheorem{remark}{Remark}[section]
\newtheorem{lema}{Lemma}[section]

\title{A canonical form for the continuous piecewise  polynomial functions}

\author{ Jorge Caravantes\thanks{Partially supported by Spanish project MTM2011-25816-C02-02}\\
Dpto. de \'Algebra \\  Universidad Complutense de Madrid, Spain\\  {\tt jorge\_caravant@mat.ucm.es}\\
M. Angeles Gomez-Molleda\\ Dpto. de \'Algebra,
Geometria y Topologia \\ Universidad de M\'alaga, Spain\\  {\tt  gomezma@agt.cie.uma.es}\\ Laureano
Gonzalez--Vega\,\,$^{*}$\,\, \\
           Dpto. de Matematicas, Estadistica y Computacion \\  Universidad de Cantabria, Spain\\  {\tt
laureano.gonzalez@unican.es} 
}

\date{ }

\begin{document}
\maketitle

\begin{abstract} We present in this paper a canonical form for the  elements in the ring of continuous piecewise
polynomial functions.  This new representation is based on the use of a  particular class of functions
$$\{C_i(P):P\in\Q[x],i=0,\ldots,\deg(P)\}$$ defined by $$C_i(P)(x)= 
\left\{ \begin{array}{cll}0 & \mbox{ if } & x \leq \alpha
\\ P(x) & \mbox{ if } & x \geq \alpha \end{array} \right.$$ where 
$\alpha$ is the $i$-th real root of the polynomial $P$. These functions will allow us to represent and manipulate easily every continuous piecewise  polynomial function through the use of the
corresponding canonical form.

It will be also shown how to produce a ``rational" representation of  each function $C_{i}(P)$  allowing its
evaluation by performing only operations in $\Q$ and avoiding the use of any real  algebraic number.

\end{abstract}

\section{Introduction} The aim of this paper is to give a canonical representation for the elements in the ring of the continuous piecewise polynomial functions. While general piecewise polynomial functions are interesting in general, most applications of them to CAGD require the functions to be continuous. In fact, splines are, by definition, sufficiently smooth piecewise-defined polynomial functions. This, then, includes the special cases b-splines and NURBS. Since some important families of curves are continuous piecewise defined polynomials, it seems useful to have a specifically defined representation for them that can take advantage of such continuity.

In \cite{VM} von Mohrenschildt proposed  a normal form, for  piecewise polynomial functions, by means of the
{\sl step} functions
$$
\step(x)=\left\{\begin{array}{lll} 1  & \mbox{ if } & x>0, \\ 0 & 
\mbox{ if } &   x\leq 0, \end{array}\right.
$$ 
which are discontinuous. {In \cite{Ch}, Chicurel-Uziel used characteristic functions of semilines to introduce a very natural form, with the same discontinuity issue. Furthermore}, Carette \cite{C} has worked on a canonical form for piecewise defined functions using range partitions. 

We are interested, however, in representing canonically the continuous piecewise polynomial functions but based on a collection of continuous functions, which should prevent errors to grow out of control when evaluating on approximate numbers. Such a suitable set of  continuous functions was introduced in \cite{D} and \cite{M}. They define, for every non-negative integer $i$, a mapping
$C_i$ from the set of polynomials on the set of continuous piecewise  polynomial functions, such that
$$C_i(P)(x)= \left\{ \begin{array}{cll}0 & \mbox{ if } & x \leq \alpha
\\ P(x) & \mbox{ if } & x \geq \alpha \end{array} \right.$$ where 
$\alpha$ is the $i$-th distinct real root of the polynomial $P$. If
$i$ is bigger than the number of real roots of $P$ then $C_i(P)$ is defined as $0$.

In \cite{D} and \cite{M} the $C_i(P)$ functions were used to study  the Pierce--Birkhoff conjecture. This is a well--known and classical open problem in Real Algebraic Geometry asking if  every  continuous and piecewise polynomial function $h\colon{\R}^n\longrightarrow{\R}$ defined over $\Q$ can be  represented by means of a sup--inf expression over a finite set of polynomials with rational coefficients. This conjecture has been  proved in the affirmative sense only for $n=1$ and $n=2$ (see \cite{D} and \cite{M}) and remains still open for $n\geq 3${, while results in \cite{B} and \cite{O} lead to a proof in certain polyhedral domains}.

In this paper we show that they provide a canonical representation which is easily computable from the piecewise expression of the functions.  Moreover performing algebraic operations between  canonical forms of continuous piecewise polynomial functions is simple and fast. In section~\ref{prel} we give the  complete definition of the $C_i$ functions along with some of their properties. Section~\ref{forma} is devoted to the proof of  existence and uniqueness of our canonical form.  In section~\ref{operaciones} we show how to obtain easily the canonical form for the sum, product and composition of continuous piecewise polynomial functions. Section~\ref{evaluacion} shows  how to produce a ``rational" representation of each function $C_{i}(P)$  allowing its evaluation by performing only  operations in $\Q$ and avoiding the use of any real algebraic number. Before the conclusions, Section~\ref{complejidad} attacks some complexity aspects of the canonical form and the operations.

\section{Preliminaries}\label{prel} Let us denote by ${\cal  CP}(\Q[x])$ the set of continuous piecewise polynomial
functions from $\R$ to $\R$ defined by polynomials with rational coefficients.

In order to represent canonically the continuous piecewise polynomial  functions,  we will use the set of mappings
$$C_i:\Q[x]
\rightarrow {\cal CP}(\Q[x]),$$ $i \in \N$,    presented in \cite{D}  and \cite{M}.

\begin{defi}\label{function}
    Let $P(x) \in \Q[x]\backslash\{0\}$, $\deg(P)=n $, $\{\alpha_{1},\ldots, 
\alpha_{r}\}$ the set of real roots of $P$,
$\alpha_{0} = -\infty $,
$\alpha_{k} = +\infty$ for every $k>r$. Then, for every $i \in \N\cup\{0\}$, $x \in \R$,
$$ C_{i}(P)(x)=\left\{ \begin{array}{ccc}0 & \mbox{\ if\ } & x\leq 
\alpha_{i} \\ P(x) &
\mbox{\ if\ } & x\geq \alpha_{i}.\end{array} \right.
$$
For completeness, we also define $C_i(P)=0$ when $P=0$.
\end{defi}

The following result can be found in \cite{D} and \cite{M} as basis  for the proof of Pierce-Birkhoff conjecture
for the case one and two dimensional. It gives a natural representation of  continuous piecewise polynomial
functions in terms of the $C_i$ functions.

\begin{prop}\label{teorLMG-V}
      Let $\phi$ be a continuous piecewise polynomial function
$$\phi(x)=\left\{ \begin{array}{lll} Q_{1}(x) &\mbox{ if }& x \leq 
\alpha_{1}\\ Q_{2}(x) &\mbox{ if }& \alpha_{1}\leq x\leq
\alpha_{2} \\
           & \vdots      & \\ Q_{N}(x) &\mbox{ if }& \alpha_{N-1}\leq  x\end{array}
   \right.$$ with $Q_i \in \Q[x]$, $Q_i\ne Q_{i+1}$ for all $i$, and $\alpha_j \in \R.$ Then $\phi$  can be written in the following way:
          $$ \phi(x)= Q_{1}(x) + 
\sum_{i=1}^{N-1}C_{s(i)}(\Delta_{i})(x)$$ where $\Delta_{i}=  Q_{i+1}-Q_{i}$ and $s(i)$ is the position index of
$\alpha_{i}$ as a root of $\Delta_{i}$.
\end{prop}

\begin{proof} For every $i\in\{1,\ldots,N-1\}$ we define the  polynomial in $\Q[x]$ given by:
$$\Delta_i(x)=Q_{i+1}(x)-Q_i(x).$$ The continuity of $\phi$ implies  that every $\alpha_i$ is a real root of
$\Delta_i(x)$. Let $s(i)$ be the position index of $\alpha_i$ as root  of $\Delta_i(x)$. In these conditions:
$$\phi=Q_1+\sum_{i=1}^{N-1}C_{s(i)}(\Delta_i)$$ as desired.
\end{proof}

\begin{ejem}\label{ejpsi}
  Let $\psi$ be defined by
     $$
     \psi(x)\!=\!\! \left\{ \!\! \begin{array}{lcc}
     x^6+1  &\mbox{ if }& x\leq \-1\cr
     x^4-\frac{1}{2}x^3-\frac{7}{2}x^2-x+6 & \!\! \mbox{ if }& \!\! \-1\leq x\leq \sqrt{2} \cr
     x^4+x^3-5x^2-4x+9 &\mbox{ if }& \sqrt{2}\leq x\end{array}
      \right.
     $$
In this case:
$$\Delta_{1}(x)= - \left( x-1 \right)  \left( {x}^{5}+\,{x}^{4}+\frac{1}{2}{x}^{2}+4\,x+
5 \right)$$
and $$\Delta_{2}(x)= \frac{3}{2}\, \left( x-1 \right)  \left( {x}^{2}-2 \right).$$
Since $1$ is the second real root of $\Delta_1(x)$ and $\sqrt{2}$ the third real root of $\Delta_2(x)$, we can write 
$$\psi(x)=x^6+1 +C_2(\Delta_1(x))+C_3(\Delta_2(x)).$$

\end{ejem} 

Proposition~\ref{teorLMG-V} along with the following properties of
$C_i$ functions will allow us to give a canonical representation of  the elements of ${\cal CP}(\Q[x])$.

\begin{prop}\label{prop1}
      Let $P(x), Q(x)$ and $H(x) \in \Q[x]$, $\alpha_1< \ldots <\alpha_r$ and
$\beta_1< \ldots < \beta_s$ the real roots of $P$ and $Q$ respectively. Then
      \begin{itemize}
      \item[i)] If $\alpha_i \geq \beta_j$ then
              $ C_{i}(P)C_{j}(Q)  = Q C_{i}(P)$.
      \item[ii)] $C_i(P^k)=P^{k-1}C_i(P)$.
      \item[iii)]  If $P = QH $ and $\alpha_i=\beta_k$ then
  $C_{i}(P)= HC_{k}(Q)$.
      \end{itemize}
\end{prop}

The proof of the proposition is straightforward.

\begin{ejem}\label{ejpsi2}
  Consider the function $\psi$ in Example~\ref{ejpsi}
      written in the form 
$$\psi(x)=x^6+1 +C_2(\Delta_1(x))+C_3(\Delta_2(x)).$$
According to Proposition~\ref{prop1}, $$C_2(\Delta_1(x))= -  \left( {x}^{5}+\,{x}^{4}+\frac{1}{2}{x}^{2}+4\,x+
5 \right)C_1\left( x-1 \right) $$ and
$$C_3(\Delta_2(x))=\frac{3}{2}\, \left( x-1 \right)  C_2\left( {x}^{2}-2 \right).$$
Therefore the pieceswise polynomial function can be written as 
 $$\psi(x)=x^6+1 -( 
{x}^{5}+{x}^{4}+\frac{1}{2}{x}^{2}+4\,x+5)C_1(x-1)+\frac{3}{2}\!\left( 
x\!-\!1 \right)\!C_2(x^2-2).$$

\end{ejem}

\section{The Canonical Form}\label{forma}

The following theorem proves the existence and uniqueness of a  canonical representation for the continuous
piecewise polynomial functions in terms of the functions $C_i$ and polynomials:

\begin{teor}\label{canonica}
      Let $\phi $: $\R\rightarrow\R $ be a continuous piecewise  polynomial function defined by polynomials in
$\Q[x]$. Then
$\phi$ can be written uniquely in the form
        $$ \phi= F_{0}+\sum_{i=1}^{N}F_{i}C_{u_{i}}(P_{i})$$ where
$F_{0}\in \Q[x]$, $F_{i}\in\Q [x]\setminus\{{0}\}$,
$P_{i} \in \Q[x] \setminus \{0\}$ is monic, irreducible, with at  least one real root  and, for every $i \in \{1,
\ldots, N\}$, the pairs
$(P_i, u_i)$ are different.
\end{teor}

\begin{proof} Let $\phi$ be a continuous piecewise polynomial function
\begin{equation}\label{StFrm}
\phi(x)=\left\{ \begin{array}{lll} Q_{1}(x) &\mbox{ if }& x \leq 
\alpha_{1}\\ Q_{2}(x) &\mbox{ if }& \alpha_{1}\leq x\leq
\alpha_{2} \\
           & \vdots      & \\ Q_{M}(x) &\mbox{ if }& \alpha_{M-1}\leq  x\end{array}
   \right.
\end{equation} 
with $Q_i \in \Q[x]$, $Q_i\ne Q_{i+1}$ for all $i$, and $\alpha_j \in \R.$

By Proposition~\ref{teorLMG-V}
$\phi$ can be written as
          $$ \phi(x)= Q_{1}(x) + 
\sum_{i=1}^{M-1}C_{s(i)}(\Delta_{i})(x)$$ where $\Delta_{i}=  Q_{i+1}-Q_{i}$ and $s(i)$ is the position index of
$\alpha_{i}$ as a root of $\Delta_{i}$.

Writing the decomposition of $\Delta_i$ into irreducible factors, 
$$\Delta_i=a_{i} f_{i_{1}}^{e_{i_1}}\ldots f_{i_{r}}^{e_{i_r}}$$   with $e_{i_j}$  $(1 \leq j \leq r)$  positive 
integers, $a_i \in \Q$ and $f_{i_1}, \ldots, f_{i_r} \in
\Q[x]$ distinct monic irreducible polynomials.

The $s(i)$--th real root of $\Delta_i$ is the $u_i$--th real root of 
$f_{i_j}$ for some $1 \leq j \leq r$. By Proposition~\ref{prop1},
$$C_{s(i)}(\Delta_i)=a_{i}f_{i_{1}}^{e_{i_1}}\ldots  f_{i_{j-1}}^{e_{i_{j-1}}}f_{i_{j+1}}^{e_{i_{j+1}}}\ldots
f_{i_{r}}^{e_{i_r}}C_{u_{i}}(f_{i_{j}}^{e_{i_j}})
  =a_{i}f_{i_{1}}^{e_{i_1}}\ldots  f_{i_{j-1}}^{e_{i_{j-1}}}f_{i_{j+1}}^{e_{i_{j+1}}}\ldots  f_{i_{r}}^{e_{i_r}}
f_{i_{j}}^{e_{i_j}-1} C_{u_{i}}(f_{i_{j}})$$

Taking $N=M-1$, $F_0=Q_1$, $P_i=f_{i_j}$ and 
$$F_i=a_{i}f_{i_{1}}^{e_{i_1}}\ldots f_{i_{j-1}}^{e_{i_{j-1}}}f_{i_{j+1}}^{e_{i_{j+1}}}\ldots  f_{i_{r}}^{e_{i_r}}
f_{i_{j}}^{e_{i_j}-1}$$ it is obtained 
\begin{equation}\label{CnFrm}
\phi(x)= F_{0}(x)+\sum_{i=1}^{N}F_{i}(x)C_{u_{i}}(P_{i})(x).
\end{equation}
If 
$k<l$ and $P_k=P_l$ then $\alpha_k$ and $\alpha_l$ are real roots of the same irreducible polynomial
$P_k$.  Since $\alpha_k < \alpha_l$ then $u_k < u_l.$

It remains to prove the uniqueness of this expression. Let us assume  that there exist two representations of
$\phi$:
\begin{equation}\label{primera}
            \phi=F_{0}+\sum_{i=1}^{N}F_{i}C_{u_{i}}(P_{i})
\end{equation}
\begin{equation}\label{segunda}
            \phi=G_{0}+\sum_{i=1}^{M}G_{i}C_{v_{i}}(R_{i})
\end{equation}

Let $\alpha_{i}$ be the $u_{i}$--th  real root of $P_{i}$ and $\beta_{i}$ the
$v_{i}$--th  real root of $R_{i}$. We can assume, without loss of  generality,  that $\alpha_{1} < \alpha_{2}
<\ldots<
\alpha_{N}$ and
$\beta_{1} <\beta_{2} <\ldots<\beta_{M}$ . Let us denote
$\theta_{i} = \min  \{{\alpha_{i} , \beta_{i}}\}$.

Applying expressions (1) and (2) to $-\infty \leq x \leq \theta_1$,
$F_0(x)=\phi(x)=G_0(x)$. Since $F_0$ and $G_0$ are polynomials, then 
$$F_0=G_0.$$

Let us assume now, without loss of generality, that $\alpha_1 \leq 
\beta_1$. If it were $\alpha_1 <\beta_1$, since
$\alpha_1 <\alpha_2$, the open interval $I=(\alpha_1, \min\{\alpha_2, 
\beta_1\})$ would not be empty and for every $x \in I$:
           $$   \begin{array}{lll}
           (\ref{primera}) &\Rightarrow & \phi(x)= F_{0}(x)+ F_{1}(x)P_{1}(x)\cr
           (\ref{segunda}) &\Rightarrow & \phi(x)=  F_{0}(x)\end{array}$$ Thus $F_{1}(x)P_{1}(x)=0$ but $F_1$ and 
$P_1$  have only finitely many roots. Therefore, $\alpha_1=\beta_1$.  Since 
$\alpha_1$ and $\beta_1$ are roots respectively of $P_1$ and
$R_1$, both monic irreducible polynomials, $$P_1=R_1 \mbox{ and } u_1=v_1.$$

Now, if $\alpha_{1}<x<\theta_{2}$ then
           $$\begin{array}{lll}
           (\ref{primera}) &\Rightarrow & \phi(x)= F_{0}(x)+ F_{1}(x)P_{1}(x)\cr
           (\ref{segunda}) &\Rightarrow & \phi(x)= F_{0}(x)+  G_{1}(x)R_{1}(x).\end{array}
           $$
           Therefore $F_{1}(x)P_{1}(x)= G_{1}(x)R_{1}(x)$. Since 
$P_1=R_1$ and it has only finitely many roots, we can conclude that $$F_1=G_1$$

  Repeating the same argument for $i=2, \ldots, \min\{N,M\}$, it is proved that
      $ F_{i} = G_{i}, \; P_{i}=H_{i}, \; u_{i}=v_{i} \mbox{ and }  N=M,$ as desired.
\end{proof}

We will call canonical form the expression obtained in the  preceding theorem.


\def\a{${\,}^{\mathrm{a}}$}
\begin{ejem}\label{ejphi} {\rm   Let us consider the following  continuous piecewise polynomial function
$$\phi(x)=\left\{ \begin{array}{llc}
x^4+4x^3-2x^2 &\mbox{\ if\ }& x\leq \alpha\cr
x^4+x^3-2x^2-3x-3 &\mbox{\ if\ }& \alpha\leq x\leq \sqrt{3} \cr
2x^3+x^2-6x-3 &\mbox{\ if\ }& \sqrt{3}\leq x
\end{array}\right.$$ where $\alpha$ is the unique real root of $x^3+x+1$. In this case
$$\Delta_{1}(x)= -3(x^3+x+1)$$ and
$$\Delta_{2}(x)= -x(x-1)(x^2-3).$$ The only real root of $\Delta_1$  is $\alpha$, while $\Delta_2$ has  roots
$-\sqrt{3}<0<1<\sqrt{3}$ so that $\sqrt{3}$ is the fourth real root  of $\Delta_2$. According to
Proposition~\ref{teorLMG-V},
$\phi(x)$ is equal to
$$x^4+4x^3-2x^2 + C_1(-3(x^3+x+1)) + C_4(-x(x-1)(x^2-3))$$ Following  the proof of Theorem~\ref{canonica} the
canonical form of $\phi(x)$ is obtained:
$$x^4+4x^3-2x^2 -3 C_1(x^3+x+1) -x (x-1) C_2(x^2-3).$$}
\end{ejem}


\section{Using the canonical form for sums, products and compositions}\label{operaciones}

In this section it is shown how to determine the canonical form for  the sum, the product and the composition of continuous piecewise polynomial functions given in canonical form. We also give an explanation on how things work for differentiation and integration. Operations with piecewise polynomial functions can be useful for computer aided geometric modeling in the sense of piecewise polynomially defined movements (or deformations) of piecewise polynomially defined objects.

The canonical  form of the sum is immediately
obtained from the canonical form of the summands. The product needs only to take into  account one of the
properties appearing in Proposition~\ref{prop1}. Composition requires considering certain real 
roots of the given polynomials and applying the corresponding rules for sums and products.

Let $\phi$ and $\psi$ be continuous piecewise polynomial functions in  canonical form, i.e.
              $$ \phi= F_{0}+\sum_{i=1}^{N}F_{i}C_{u_{i}}(P_{i})$$
              $$ \psi= G_{0}+\sum_{j=1}^{M}G_{j}C_{v_{j}}(R_{j})$$

\subsection{Sums and products} \label{suma} If the pairs $(P_i,u_i), 
\; 1\leq i \leq N, \; (R_j, v_j), \; 1 \leq j \leq M$ are all different then the canonical form of $\phi+ \psi$ is
clearly
$$(F_0+G_0)+\sum_{i=1}^{N}  F_iC_{u_{i}}(P_i)+\sum_{j=1}^{M}G_jC_{v_{j}}(R_j).$$ Otherwise, it  will be enough to
sum up the terms
$F_iC_{u_{i}}(P_i),   R_jC_{v_{j}}(R_j)$ such that 
$(P_i,u_i)=(R_j,v_j)$ in order to obtain the canonical form.

  Concerning the product $\phi\psi$, we only need to obtain the  canonical form of the terms of type
   $$F_i C_{u_{i}}(P_i) G_j C_{v_{j}}(R_j)=F_i G_j C_{u_{i}}(P_i)  C_{v_{j}}(R_j).$$ We can assume, without loss of
generality, that the
$u_{i}$--th real root of
$P_i$ is greater than or equal to the $v_{j}$--th real root of $R_{j}$.  Then, by Proposition~\ref{prop1}
$$F_{i} G_j C_{u_{i}}(P_i) C_{v_{j}}(R_j)=F_iG_j R_j C_{u_{i}}(P_i).$$

\begin{ejem}
Let us determine the canonical form for the product 
of $\phi$ and  $\psi$,  the functions considered in
examples~\ref{ejphi} and ~\ref{ejpsi} respectively.
We compute the product of the canonical form of $\phi$ and $\psi$ 
term by term,  taking into account that
$$C_1(x^3+x+1)C_1(x-1)=(x^3+x+1)C_1(x-1),$$ since the only real root 
of $x^3+x+1$ is smaller than $1$. By the same
reasoning,
$$C_1(x^3+x+1)C_2(x^2-2)=(x^3+x+1)C_2(x^2-2),$$
$$C_2(x^2-3)C_1(x-1)=(x-1)C_2(x^2-3)$$  and 
$$C_2(x^2-3)C_2(x^2-2)=(x^2-2)C_2(x^2-3).$$ Therefore, the canonical 
form of
$\phi \psi(x)$ is
$$H_0+H_1C_1(x^3+x+1)+H_2C_1(x-1)+H_3C_2(x^2-2)+H_4C_2(x^2-3)$$ where
$$H_0={x}^{10}+4\,{x}^{9}-2\,{x}^{8}+{x}^{4}+4\,{x}^{3}-2\,{x}^{2},$$
$$H_1=-3\,{x}^{6}-3,$$
$$H_2=-{x}^{9}-2\,{x}^{8}+{x}^{7}+\frac{9}{2}\,{x}^{6}+\frac{3}{2}\,{x}^{5}-5\,{x}^{4}+\frac{9}{2}\, 
{x}^{3} +{\frac
{47}{2}}\,{x}^{2}+27\,x+15,$$
$$H_3=3/2\,{x}^{5}-9/2\,{x}^{3}-3/2\,{x}^{2}+9/2$$ and
$$H_4=-{x}^{6}+6\,{x}^{4}-{x}^{3}-13\,{x}^{2}+9\,x.$$ 
\end{ejem}

\subsection{Compositions} \label{composicion}

We are interested in determining the canonical form of $\phi \circ 
\psi$.  Since we have already seen how to compute the canonical form of sums and products, the only remaining
problem  concerning the composition is the computation of the canonical form of the expressions
$$C_{v_j}(R_j)(\phi(x))=
\left\{ \begin{array}{lcc} 0 &\mbox{ if }& \phi(x)\leq
\alpha_{j} \\ R_{v_j}(\phi(x))  \hfill  &\mbox{ if }& \phi(x)\geq 
\alpha_{j}\end{array} \right.
$$ where $\alpha_j$ is the $v_j$--th real root of $R_j$.

First we must determine the values of $x$ for which $\phi(x) \geq 
\alpha_j$.  Since $\phi$ is a continuous function, it is enough to find the points where
$\phi$ takes the value $\alpha_j$ and then give the sign of 
$\phi(x)-\alpha_j$ in each interval. The points where $\phi(x)$ is equal to $\alpha_j$ are real roots of $R_j
\circ \phi$.  Therefore, we can determine the real roots of $R_j \circ \phi$, which is a piecewise polynomial
function, and take those roots 
$\gamma_1, \ldots, \gamma_s$  such that
$\phi(\gamma_k)=\alpha_j$. Let $\gamma_0=-\infty,$ $\gamma_{s+1}=\infty.$ Thus,
$$C_{v_j}(R_j)(\phi(x))=
\left\{ \begin{array}{lcc} 0 &\mbox{ if }& x \in I_1\\  R_{v_j}(\phi(x))  \hfill  &\mbox{ if }& x \in
I_{2},\end{array} 
\right.
$$ where $I_1$ is the union of the intervals $(\gamma_i, 
\gamma_{i+1})$ such that
$\phi(x)\leq \alpha_j$ and $I_2=\R \setminus I_1.$

Partitioning $I_2$ by the intervals of definition of $\phi$, we obtain
the usual form of a piecewise polynomial function and, by 
Theorem~\ref{canonica},  the canonical form is easily computable.

\begin{ejem}{\rm
Let us compute the canonical form of 
$C_2(x^2-2)(\phi(x))$, being $\phi$ the function defined in  Example~\ref{ejphi}.  First we must determine the values of $x$  for which $\phi(x) \leq \sqrt{2}.$ To this purpose, we have computed the roots $\gamma$  of the function
$$(x^2-2)\circ \phi(x)=
\left\{ \begin{array}{llc}
\!\!(x^2-2)\circ (x^4+4x^3-2x^2) & \!\! \mbox{\ if\ }& \! x\leq \alpha\\
\!\!(x^2-2)\circ(x^4+x^3-2x^2-3x-3) &\!\!\mbox{\ if\ }&\!\! \alpha\leq x\leq \sqrt{3} \\
\!\!(x^2-2)\circ(2x^3+x^2-6x-3) &\!\!\mbox{\ if\ }& \sqrt{3}\leq  x
\end{array}\right.
$$ 
such that $\phi(\gamma)=\sqrt{2}$.

The roots of $(x^2-2)\circ \phi(x)$ are:
\begin{itemize}
\item The roots of $(x^2-2)\circ(x^4+4x^3-2x^2)$ in $(-\infty,\alpha]$,
\item The roots of $(x^2-2)\circ(x^4+x^3-2x^2-3x-3)$ in $[\alpha,\sqrt{3})$, and
\item The roots of $(x^2-2)\circ(2x^3+x^2-6x-3)$ in $[\sqrt{3},\infty)$.
\end{itemize}
One can compare $\alpha$ and the roots of these three polynomials:
\begin{itemize}
\item The first and second real roots of $(x^2-2)\circ(x^4+4x^3-2x^2)$ are less than $\alpha$. Applying $\phi$ gives $\sqrt{2}$ for the first one and $-\sqrt{2}$ for the second one.
\item The third real root of $(x^2-2)\circ(x^4+x^3-2x^2-3x-3)$ is the only one greater than $\alpha$ and less than $\sqrt{3}$, but its image by $\phi$ is $-\sqrt{2}$.
\item The sixth real root of $(x^2-2)\circ(2x^3+x^2-6x-3)$ is the only one greater than $\sqrt{3}$, and its image by $\phi$ is $\sqrt{2}$.
\end{itemize}
Therefore, $\gamma_1$, the first real root of $(x^2-2)\circ(x^4+4x^3-2x^2)$, and $\gamma_2$, the sixth real root of $(x^2-2)\circ(2x^3+x^2-6x-3)$ are the algebraic numbers we looked for. Then we compute that
$$C_2(x^2-2)(\phi(x))=\left\{ \begin{array}{llc}
(x^2-2)\circ \phi(x) &\mbox{\ if\ }& x\leq \gamma_1\cr
0 &\mbox{\ if\ }& \gamma_1\leq x\leq \gamma_2 \cr
(x^2-2)\circ \phi(x) &\mbox{\ if\ }& \gamma_2 \leq x
\end{array}\right.$$  
Since $\gamma_1 \leq \alpha \leq \sqrt{3} \leq \gamma_2$, we have that 
$$C_2(x^2-2)(\phi(x))=\left\{ \begin{array}{llc}
(x^2-2)\circ (x^4+4x^3-2x^2) &\mbox{\ if\ }& x\leq \gamma_1\cr
0 &\mbox{\ if\ }& \gamma_1\leq x\leq \gamma_2 \cr
(x^2-2)\circ(2x^3+x^2-6x-3) &\mbox{\ if\ }& \gamma_2 \leq x
\end{array}\right.$$

Now applying Theorem~\ref{canonica} we obtain the canonical form
$$C_2(x^2-2)(\phi(x))=S_0-C_1(S_0)+C_6(S_2),$$ where $$S_0=(x^2-2)\circ
(x^4+4x^3-2x^2)={x}^{8}+8\,{x}^{7}+12\,{x}^{6}-16\,{x}^{5}+4\,{x}^{4}-2$$ and
$$S_2=(x^2-2)\circ(2x^3+x^2-6x-3)  =4\,{x}^{6}+4\,{x}^{5}-23\,{x}^{4}-24\,{x}^{3}+30\,{x}^{2}+36\,x+7,$$  
both irreducible polynomials.}
\end{ejem}

\subsection{A note on differentiation and integration}

One could think of integration and differentiation as interesting operations on continuous functions. 
While the derivative of a continuous piecewise polynomial function does, in general, 
not exist at the endpoints of the defining intervals, one can easily determine when this case happens.

\begin{lema} Let us consider a continuous piecesewise polynomial function $\phi$ in its canonical form
\[ \phi= F_{0}+\sum_{i=1}^{N}F_{i}C_{u_{i}}(P_{i}).\]

The derivative $\phi'$ is continuous if and only if $P_i$ divides $F_i$ for every $i=1,...,N$. In this case,
\[ \phi'= F'_{0}+\sum_{i=1}^{N}\left(F_{i}'+\frac{F_{i}}{P_{i}}P_{i}'\right)C_{u_{i}}(P_{i}).\]
\end{lema}
\begin{proof}
  Let us call $\alpha_i$ the $u_i-th$ real root of $P_i$.
The function $\phi$ is differentiable at the point $\alpha_i$ if and only if the left and right derivatives are equal at $\alpha_i$, i.e. 
$$\left(F_0+\sum_{j=1}^{i-1}F_jP_j\right)'(\alpha_i)= \left(F_0+\sum_{j=1}^{i}F_jP_j\right)'(\alpha_i).$$
This equality holds true if and only if $\left(F_iP_i\right)'(\alpha_i)=0$. Since $F_iP_i(\alpha_i)=0$, we have that $\alpha_i$ is a multiple root of
$F_iP_i$. Taking into account that $P_i$ is irreducible and has $\alpha_i$ as a root, we conclude that $\phi$ is differentiable at
$\alpha_i$ if and only if $P_i$ divides $F_i$
in $\Q[x].$   
 When $P_i$ divides $F_i$ for every $i=1,...,N$, the derivative $\phi'$ is clearly continuous and, noticing that 
 $$C_{v_i}\left((F_iP_i)'\right)=\left(F_i'+\frac{F_i}{P_i}P_i'\right)C_{u_i}(P_i),$$ 
 where $v_i$ is the position of $\alpha_i$ as real root of $(F_i P_i)',$ we have that 
 $$ \phi'= F'_{0}+\sum_{i=1}^{N}\left(F_{i}'+\frac{F_{i}}{P_{i}}P_{i}'\right)C_{u_{i}}(P_{i}).$$
\end{proof}

Regarding integration, it has been treated for piecewise functions in \cite{JLMR} and \cite{JR}. However, while the primitive of a continuous piecewise polynomial function remains continuous, the field of definition of the polynomial must, in general, be extended. Think, for example of the function $\psi=3C_2(x^2-2)$, i.e.
\[\psi(x)= \left\{ \begin{array}{cll}0 & \mbox{ if } & x \leq \sqrt{2}
\\ 3x^2-6 & \mbox{ if } & x \geq \sqrt{2} \end{array} \right.\]
Then any of its continuous primitives has the shape:
\[\phi(x)= \left\{ \begin{array}{cll} 
a-4\sqrt{2} & \mbox{ if }  & x \leq \sqrt{2},\\ 
x^3-6x+a & \mbox{ if } & x \geq \sqrt{2}, 
\end{array} \right.\]
where $a\in \R$. It is obvious that at least one of the two pieces of $F$ is defined by a polynomial that is not in $\Q[x]$, 
so $\phi\not\in{\cal CP}\left(\Q[x]\right)$.

\section{The evaluation of $C_i(P)$}\label{evaluacion}

{\rm  
In this section we introduce two methods of evaluation of $C_i(P)$.
}

\subsection{The simplest alternative}

{\rm  
Perhaps the most immediate way to evaluate $C_i(P)$ could be the following one. We have $a_i$ and $b_i$ from the moment of the definition of $C_i(P)$. Then:
\begin{itemize}
\item If $x<a_i$, then $C_i(P)(x)=0$.
\item If $x>b_i$, then $C_i(P)(x)=P(x)$.
\item If $x\in[a_i,b_i]$, since $P$ is irreducible and monic, it is easy to check that the sign of $P$ in $\alpha_i$ is sign$_{a_i}=(-1)^{i+\deg(P)-1}$ and the sign in $b_i$ is sign$_{b_i}=(-1)^{i+\deg(P)}$. So, if sign$(P(x))=$sign$_{b_i}$, then $C_i(P)(x)=P(x)$. Otherwise $C_i(P(x))=0$.
\end{itemize}
}

\subsection{A closed form solution}

{\rm  
While the former method seems quite simple and fast, it is worth introducing a different one with the appeal of a closed formula and gives a constructive approach to Pierce--Birkhoff conjecture.
}

In this subsection we show that every $C_i(P)$ can be evaluated at any  rational (or real value) 
with just a rational isolating interval $[a_i,b_i]$ of the $i$--th real root $\alpha_i$ of $P$ (i.e. $\alpha_{i-1}<a_i\le\alpha_i\le b_i<\alpha_{i+1},$ $a_i, b_i \in \Q$). {Observe that, if $\alpha_i$ is rational (and then $P$ is linear), the evaluation problem is trivial, so we will suppose this is not the case and then we have an isolating interval of positive length.}
For that we show that each function $C_i(P)(x)$ has a rational expression in terms of the polynomial $P$, its absolute value $|P|$
and very simple piecewise polynomial functions in the canonical basis whose defining intervals have rational endpoints.  
In fact, 
we prove that each function $C_i(P)(x)$ is a semipolynomial over $\Q$, 
i.e. a function 
that  can be described as the
composition of polynomials in
$\Q[x]$ and the absolute value function $\vert\bullet\vert$.
This representation provides a new algorithm for evaluating $C_i(P)$ that avoids
 the inherent exponential complexity of the algorithm in \cite{D} or \cite{M} . Anyway it is still  open if the
algorithm presented in this section has or not a polynomial complexity while the first performed experiments show a
very good practical behaviour.

The procedure to express $C_i(P)$ as a semipolynomial is based on a hint a referee gave to us: Given an isolating interval, the expression
\[
C_i(P)(x)=\left\{
\begin{array}{ll}
0&\mbox{ if }x\le a_i\\
0&\mbox{ if }a_i\le x\le b_i\mbox{ and } sign P(x)=sign P(a_i)\\
P(x)&\mbox{ if }a_i\le x\le b_i\mbox{ and }  sign P(x)=sign P(b_i)\\
P(x)&\mbox{ if }x\ge b_i
\end{array}
\right.
\]
allows to easily evaluate $C_i(P)$.


To obtain the expression of $C_i(P)$ as a semipolynomial, let us denote
$n=\deg(P)$,
$\{\alpha_1,\ldots,\alpha_r\}=\{\alpha\in \R: P(\alpha)=0\}$, $\alpha_0=-\infty$ and $\alpha_{r+1}=+\infty$. Firstly, we only need to study the case  where $P$ is squarefree because
$$\widetilde{P}=\frac{P}{\gcd(P,P^{(1)})}\quad\Longrightarrow\quad  C_i(P)= \gcd(P,P^{(1)})C_i(\widetilde{P})$$
Also we can assume that $i\in\{1,\ldots,r\}$ because 
$$C_0(P)=P\quad\quad \mbox{ and }\quad\quad C_{r+1}(P)=0.$$

Now we define the following semipolynomials:
$$[P]^+=\sup\{P,0\}=\frac{P+\vert P\vert}{2},$$
$$[P]^-=\inf\{P,0\}=\frac{P-\vert P\vert}{2}.$$
In the easiest case: $a_i=\alpha_i$ or $b_i=\alpha_i$ (just checked by computing $P(a_i)$ and $P(b_i)$), we know $\alpha_i$, which is rational and
$$P(x)=(x-\alpha_i)Q(x).$$
In this case, $C_i(P)(x)=[x-\alpha_i]^+Q(x)$. 

In general, we have $a_i<\alpha_i<b_i$
so that we can choose any $a'\in[a_i,\alpha_i)$ and $b'\in(\alpha_i,b_i]$. Let us define the semipolynomials
\[
f(x):=\left[\frac{x-a_i}{a'-a_i}\right]^+-\left[\frac{x-a'}{a'-a_i}\right]^+=\left\{\begin{array}{ll}
0&\mbox{ if }x<a_i\\
\displaystyle\frac{x-a_i}{a'-a_i}&\mbox{ if }a_i\le x\le a'\\
1&\mbox{ if }a'<x\\
\end{array}\right.
\]
\[
g(x):=\left[\frac{x-b'}{b_i-b'}\right]^+-\left[\frac{x-b_i}{b_i-b'}\right]^+=\left\{\begin{array}{ll}
0&\mbox{ if }x<b'\\
\displaystyle\frac{x-b'}{b_i-b'}&\mbox{ if }b'\le x\le b_i\\
1&\mbox{ if }b_i<x\\
\end{array}\right.
\]
which are continuous piecewise polynomial functions with canonical forms
$$f(x)=\frac{1}{a'-a_i}\left(C_1(x-a_i)-C_1(x-a') \right)$$
$$g(x)=\frac{1}{b_i-b'}\left(C_1(x-b')-C_1(x-b_i) \right)$$

We will use $[P]^+, [P]^-, f$ and $g$ to prove the following result:

\begin{teor}\label{nuevo}
Let $P$ be a squarefree polynomial in $\Q[x]$ of degree $n$ with $r$ real roots $\alpha_1<...<\alpha_r$. Let $a_i$, $a'$, $b_i$ and $b'$ be real numbers such that $\alpha_{i-1}<a_i<a'<\alpha_i<b'<b_i<\alpha_{i+1}$. Then
$$C_i(P)(x)=\frac{1}{b_i-b'}[{P}]^{signP(a_i)}(x)\left(C_1(x-b')-C_1(x-b_i) \right)+ \frac{1}{a'-a_i}[{P}]^{signP(b_i)}(x)\left(C_1(x-a_i)-C_1(x-a') \right)$$

\end{teor}
\begin{proof}
We can write the expression on the right side of the equation as 
$[{P}]^{signP(a_i)}(x)g(x)+[{P}]^{signP(b_i)}(x)f(x)$.
 Since $f$ and $g$ are equal in $(-\infty, a_i), (b_i, \infty)$ and $[{P}]^++[{P}]^-=P$, it is obvious that 
 $[{P}]^{signP(a_i)}g+[{P}]^{signP(b_i)}f$ is $0$ and $P$ respectively in $(-\infty, a_i)$ and  $(b_i, \infty).$ It is also clearly $0$ at $\alpha_i.$ 
 It remains to study this function in the intervals $[a_i,\alpha_i)$ and $(\alpha_i, b_i].$ Since $signP(x)=signP(a_i)$ for every $x \in [a_i,\alpha_i)$ and
 $signP(x)=signP(b_i)$ for every $x \in (\alpha_i,b_i]$ we have:
 $$\mbox{ if } signP(a_i)=+ \mbox{ then } [{P}]^{+}g+[{P}]^{-}f=\left\{\begin{array}{ccc}0 & \mbox{in} & [a_i,\alpha_i)\\P & \mbox{in} & (\alpha_i,b_i] \end{array}\right.$$
 $$\mbox{ if } signP(a_i)=- \mbox{ then } [{P}]^{-}g+[{P}]^{+}f=\left\{\begin{array}{ccc}0 & \mbox{in} & [a_i,\alpha_i)\\P & \mbox{in} & (\alpha_i,b_i] \end{array}\right.$$
 We conclude that $[{P}]^{signP(a_i)}g+[{P}]^{signP(b_i)}f=C_i(P).$ 
\end{proof}

\begin{remark}
Since the isolating interval of the $i$-th root is computed when passing a piecewise polynomial function to canonical form, the only extra time will be choosing $a'\in[a_i,\alpha_i)$ and $b'\in(\alpha_i,b_i]$.
\end{remark}

\begin{ejem}\label{ejemplonuevo}{\rm  
To compute the above expressions for the functions $\{C_{i}(P)\}_i$ for $P=x^3-3x+1$, we first realize that $-3<-2<\alpha_1<-1<0<\alpha_2<\frac{1}{2}<1<\alpha_3<2<3$, where the $\alpha_i$ are the roots of $P$ (just by checking the sign of $P(x)$ for $x=-3,-2,-1,0,\frac{1}{2},1,2,3$). 
Then by Theorem \ref{nuevo}, it is obvious that:
\[
C_1(P)(x)=[{P}]^-(x)\left(C_1(x+1) -C_1(x)\right)+ [{P}]^+(x)+\left(C_1(x+3) -C_1(x+2)\right),
\]
\[
C_2(P)(x)=2[{P}]^+(x)\left(C_1\left(x-\frac{1}{2}\right) -C_1(x-1)\right)+ [{P}]^-(x)\left(C_1(x+1) -C_1(x)\right),
\]
\[
C_3(P)(x)=[{P}]^-(x)\left(C_1(x-2)-C_1(x-3)\right)+2[{P}]^+ (x)\left(C_1\left(x-\frac{1}{2}\right) -C_1(x-1)\right).
\]
}\end{ejem}

\section{On complexity}\label{complejidad}

{{\rm
\begin{remark}
The canonical form in this paper should be, in general, lighter than the piecewise standard one and those in \cite{VM} and \cite{C}. The point is that the classical methods to represent piecewise polynomial functions store the polynomials that define the function in the different intervals (the $Q_i$ in (\ref{StFrm})) and the extremes of such intervals (the roots $\alpha_j$), which could be algebraic numbers (and then their minimal polynomial must also be kept). In this new canonical form, one still has to store one polynomial for every interval (all $F_i$ in (\ref{CnFrm})) and also one polynomial for every extreme (the $P_i$). However, 
\[Q_j(x)=F_{0}(x)+\sum_{i=1}^{j}F_{i}(x)P_{i}(x),\]
which is expected to have greater degree than $F_i$ in general. So the representation and storage complexity of the piecewise polynomial functions when considered in canonical form (i.e. storing all $F_i$ and all $P_i$ with the intervals separating $\alpha_i$ from the other roots) should be less than that when considered in standard form, since both the bigger $Q_i$ and the $\alpha_i$, whose information include $P_i$ and the isolating interval, must be stored.

As an example, the piecewise polynomial function in Example \ref{ejpsi}, is 145 bytes long in Maple (from the command "{\rm length}") while the canonical form in sage returns 72 for Python's command "{\rm getsizeof}".
\end{remark}
}}

{{
When comparing with other forms, the sum with the canonical form introduced in this paper is clearly easier to manage, since we just check equalities of monic polynomials and then we either sum polynomials or add terms to a list. 
}}

{{
Given $\phi$ and $\psi$ piecewise polynomial continuous functions in canonical form with $N$ and $M$ summands each, their product takes $(N-1)(M-1)$ comparisons of algebraic numbers, $2NM$ products of polynomials and $MN-M-N+1$ sums of polynomials (since, at most, the product of two piecewise functions has as many breakpoints as the two functions together: $N+M-1$). 
}}

{{Composition $\phi\psi$ involves the resolution of $O(MN)$ polynomial equations ($M$ for each summand of $\phi$), a varying number (depending on the number of real roots of the polynomials, $N$ and $M$) of algebraic number comparisons, up to $O(NM)$ compositions of polynomials (again $M$ for each summand of $\phi$). The complexity of the result depends more on the number of real roots of the polynomials we use than on $N$ or $M$.
}}

{{
The first method of evaluation, when computing $\phi(x)$ needs to compare $x$ with $N-1$ algebraic numbers and then we need to evaluate up to $N$ polynomials and sum them up to $N$ resulting numbers.
}}

{{
Finally, for the second method of evaluation, we can assume that 
\[a',\ b',\ \frac{1}{b_i-b'},\mbox{ and } \frac{1}{a_i-a'}\]
 are already computed (observe that $a_i$ and $b_i$ are already in the description of the algebraic number $\alpha_i$), since it can be done just once for all evaluations (in fact, they can be gotten when the canonical form is computed, and storing them, as rationals, is not too expensive). Then we compute $\phi(x)$ just by evaluating $N$ rational (with the license of taking absolute values) functions and summing the results.
}}

{{\rm
\begin{remark}
If we want to preserve an order in the summands (for instance we could order them by the order of the roots $u_i$ of the polynomials $P_j$), we would complicate the sum with some comparisons of algebraic numbers. On the other side, we would remove comparisons from product (reduced to $M+N-2$ by using the order among the roots) and evaluation (reduced $O(\log_2(N))$ if we use a binary search as in \cite{C}). Moreover, composition would also be simplified if we store the order data of the roots.
\end{remark}
}}

{{\rm
\begin{ejem}
The canonical form has been implemented in Sage as a class\footnote{The source code is available at  
\url{http://www.mat.ucm.es/~jcaravan/Paquetillos/CPWPF_SAGE02.sage}} and tested running on an AMD Turion 64$\times2$ at 1,9 GHz with 2GB RAM under Ubuntu 14.04. Given the functions:
\[\begin{array}{l}
\displaystyle\phi(x)=x^3 - 5 +C_1(x^2 - 2)+C_2(x^2 - 2)+(x^3 - 2x + 1)C_1(x^3 - x - 7)\\
\displaystyle\psi(x)=x^6 + 1 -\left(x^5 + x^4 + \frac{1}{2}x^2 + 4x + 5\right)C_1(x - 1) +\left(\frac{3}{2}x - \frac{3}{2}\right)C_2(x^2 - 2),
\end{array}\]
the computer gave the sum
in 4.51 ms. It computed the product
in 84.6 ms. Finally, the composition
took 709 ms and the evaluation of such composition at $5/4$ took 238 $\mu$s.
\end{ejem}
}}

\section{Conclusions} We have introduced a new canonical form for  the elements in the ring of the continuous piecewise 
polynomial functions from $\R$ to $\R$ defined by polynomials in $\Q[x]$. In algebraic  terms, we have shown that this ring agrees
with the
$\Q[x]$--module generated by the functions:
$$\begin{array}{l}\{C_{i}(P):i\in\N, P\in \Q[x] \setminus 
\{0\}\;\hbox{monic and irreducible}\}.\end{array}$$ It has been also shown  how to use
this  canonical form in order to perform sums, products and compositions and to obtain  the corrresponding
canonical form of the result.

We have also presented how to produce a ``rational" representation of  the $C_{i}(P)$ functions allowing its evaluation by performing only operations in $\Q$ and avoiding the use of 
real  algebraic numbers.



To finish, all the obtained results and algorithms can be stated in terms of an  ordered field $\K$ (instead of $\Q$) and a
real--closed field $\F$ (instead of $\R$) containing $\K$. In this case, and when  the real--closed $\F$ is not
archimedean preventing from using isolating intervals, the manipulation of the involved real  algebraic numbers in
the proof of Theorem
\ref{nuevo} must be performed by using Thom's codes (see \cite{BPR}).


%
%
\end{document}